\documentclass[conference]{IEEEtran}
\IEEEoverridecommandlockouts

\title{Dynamic-Aware LC-DMA Beam Tracking for \newline LEO Satellite Communications}
\author{
\IEEEauthorblockN{Mohamadreza Delbari, Mohammad Zarei, Denis Giniatoulline, and Vahid Jamali}
\IEEEauthorblockA{
Technical University of Darmstadt, Darmstadt, Germany 
\thanks{This work was supported in part by the DFG, German Research Foundation, within the Collaborative Research Center MAKI (SFB 1053) under Project-ID 210487104; in part by the LOEWE initiative (Hesse, Germany) within the emergenCITY Centre under Grant LOEWE/1/12/519/03/05.001(0016)/72, and in part by the German Federal Ministry for Research, Technology and Space (BMFTR) under the program of ``Souverän. Digital. Vernetzt.'' joint project Open6GHub plus (Project-ID 16KIS2407).}
}
}

\usepackage{amsfonts,amsmath,amsthm,bm}
\usepackage{cite}
\usepackage{amsmath,amssymb,amsfonts}
\usepackage{algorithmic}
\usepackage{graphicx}
\usepackage{textcomp}
\usepackage{xcolor}
\usepackage[T1]{fontenc}
\usepackage{comment}
\usepackage{eucal}
\usepackage{algorithm}
\usepackage{algorithmic}
\usepackage[acronyms,nonumberlist,nopostdot,nomain,nogroupskip]{glossaries}
\usepackage{comment}
\usepackage[font=footnotesize]{caption}
\usepackage{subcaption}
\usepackage{balance}
\usepackage{flafter}

\def\bGamma{\boldsymbol{\Gamma}}
\def\btheta{\boldsymbol{\theta}}
\def\bomega{\boldsymbol{\omega}}
\def\bvarphi{\boldsymbol{\varphi}}
\def\bPsi{\boldsymbol{\Psi}}
\def\balpha{\boldsymbol{\alpha}}
\def\bbeta{\boldsymbol{\beta}}
\def\bzeta{\boldsymbol{\zeta}}
\def\bmu{\boldsymbol{\mu}}
\def\bSigma{\boldsymbol{\Sigma}}
\def\blambda{\boldsymbol{\lambda}}
\def\bzero{\boldsymbol{0}}
\def\bone{\boldsymbol{1}}

\def\ba{\mathbf{a}}
\def\bI{\mathbf{I}}
\def\bH{\mathbf{H}}
\def\bx{\mathbf{x}}
\def\bX{\mathbf{X}}

\def\by{\mathbf{y}}
\def\bn{\mathbf{n}}
\def\bu{\mathbf{u}}
\def\bq{\mathbf{q}}
\def\bt{\mathbf{t}}
\def\bw{\mathbf{w}}
\def\bg{\mathbf{g}}
\def\bG{\mathbf{G}}
\def\bm{\mathbf{m}}
\def\bQ{\mathbf{Q}}
\def\br{\mathbf{r}}
\def\bL{\mathbf{L}}
\def\be{\mathbf{e}}
\def\bh{\mathbf{h}}

\def\tx{\mathrm{tx}}
\def\rx{\mathrm{rx}}
\def\tmax{\mathrm{max}}
\def\SNR{\mathrm{SNR}}
\def\thr{\mathrm{thr}}
\def\diag{\mathrm{diag}}
\def\Cset{\mathbb{C}}
\def\Rset{\mathbb{R}}
\def\Ex{\mathbb{E}}
\def\sCN{\mathcal{CN}}

\def\Trans{\mathsf{T}}
\def\Herm{\mathsf{H}}
\newcommand{\e}{\mathsf{e}}
\newcommand{\jj}{\mathsf{j}}

\newcommand{\defeq}{\triangleq}

\newtheorem{lem}{Lemma}

\newacronym{LEO}{LEO}{low Earth orbit}
\newacronym{LC}{LC}{liquid crystal}
\newacronym{LoS}{LoS}{line-of-sight}
\newacronym{Rx}{Rx}{receiver}
\newacronym{Tx}{Tx}{transmitter}
\newacronym{AWGN}{AWGN}{additive white Gaussian noise}
\newacronym{DMA}{DMA}{dynamic metasurface antenna}
\newacronym{RF}{RF}{radio frequency}
\newacronym{UPA}{UPA}{uniform planar array}
\newacronym{AoD}{AoD}{angle of departure}
\newacronym{AoA}{AoA}{angle of arrival}
\newacronym{QoS}{QoS}{quality of service}
\newacronym{SNR}{SNR}{signal-to-noise ratio}
\newacronym{AO}{AO}{alternative optimization}
\newacronym{MRT}{MRT}{maximum ratio transmission}
\newacronym{MRC}{MRC}{maximum ratio combining}
\newacronym{NTN}{NTN}{non-terrestrial network}
\newacronym{DC}{DC}{direct current}
\begin{document}
\maketitle

\begin{abstract}
\Gls{LC}-based \glspl{DMA} offer a power-efficient, low-cost, and scalable solution for \gls{LEO} satellite communications where large apertures are needed to compensate for large path loss in the Earth-satellite link. 
However, beam tracking with \gls{LC}-\glspl{DMA} is governed by two coupled dynamics: the required time-varying beam direction caused by user and satellite mobility, and the finite response time of the \gls{LC} phase shifters. This raises the question of whether \gls{LC}-\glspl{DMA} can reconfigure sufficiently fast to support high-mobility \gls{LEO} satellite links. To address this problem, we develop a dynamic-aware beam-tracking framework that explicitly accounts for both dynamics. In particular, we formulate an optimization problem to maximize the cumulative received \gls{SNR} over a predictive tracking horizon subject to the physical dynamics of the \gls{LC} unit cells as well as transmit power constraints. To solve this non-convex problem efficiently, we develop an alternating optimization framework combining a closed-form \gls{MRT} digital precoder and \gls{MRC} with a parallel, closed-form dynamic-aware phase-shift update. Simulation results show that when the time response of \gls{LC}-\glspl{DMA} is properly accounted for, it can support practical mobility scenarios. Moreover, the proposed dynamic-aware algorithm maintains stable \gls{SNR} and outperforms conventional designs, particularly at high vehicle velocities and low orbital azimuth angles.
\end{abstract}

\glsresetall
\section{Introduction}
\Glspl{NTN} utilizing \gls{LEO} satellites have emerged as a key component for 6G wireless connectivity and providing continuous global coverage to remote areas and mobile devices \cite{su2019broadband,kodheli2020satellite}. Nevertheless, establishing reliable uplink communication between \gls{LEO} satellites and high-speed vehicles presents several challenges. In particular, high beamforming gains are required to compensate for the large path loss, and adaptive beamforming with continuous beam tracking is needed to accommodate rapid \gls{LEO} satellite motion and vehicle mobility \cite{you2020massive}.

Conventional phased arrays rely on numerous high-resolution phase shifters and dedicated power amplifiers. This leads to prohibitive power consumption and hardware cost when scaled to large aperture sizes \cite{Shlezinger2019}. To overcome these limitations, \glspl{DMA} are considered as an energy-efficient alternative for holographic beamforming \cite{smith2017}. A \gls{DMA} integrates dense sub-wavelength metamaterial radiating elements directly onto microstrip feed lines driven by a reduced number of \gls{RF} chains compared to the conventional phased arrays. By dynamically tuning the Lorentzian resonance of individual unit cells, \glspl{DMA} achieve flexible spatial multiplexing and beam focusing directly in the electromagnetic domain and at the same time substantially reduce the RF hardware overhead \cite{smith2017,Kimaryo2023}. 

Among various reconfigurable technologies (e.g., PIN diodes and varactors), \gls{LC} technology provides continuous phase tunability, high operating frequencies (into millimeter-wave and terahertz regimes), and near-zero static \gls{DC} power consumption \cite{Yoo2025,Delbari2026,Xu2024}. In an LC-based DMA, applying a low-frequency bias voltage rotates the anisotropic nematic \gls{LC} molecules, which alters the effective local dielectric permittivity and shifts the resonance frequency of the metamaterial element \cite{smith2017,Park2024}. These properties make \gls{LC}-\glspl{DMA} an attractive choice for large-aperture, energy-efficient satellite terminals.

Nevertheless, beam tracking with \gls{LC}-\glspl{DMA} is governed by two coupled dynamics: the mobility-induced dynamics of the vehicle-to-satellite link and the internal dynamics of the \gls{LC} phase shifters. The former continuously changes the required beam direction due to vehicle and satellite mobility, while the latter limits how fast the beam can be reconfigured. In particular, the molecular reorientation speed of \gls{LC} mixtures is fundamentally governed by rotational viscosity and elastic restoration forces, which results in switching response times on the order of milliseconds \cite{Delbari2024}. This mismatch between the required and achievable beam dynamics may become a performance bottleneck in \gls{LEO} satellite systems. Conventional beam-tracking algorithms assume instantaneous phase transitions and aim to maximize the instantaneous \gls{SNR}. These algorithms may demand large phase changes across consecutive tracking steps that cannot be realized instantaneously with \gls{LC}-\glspl{DMA}, and hence result in transition periods during which the received \gls{SNR} can severely degrade. This interplay between mobility and hardware dynamics raises the question of whether \gls{LC}-\glspl{DMA} can reconfigure sufficiently fast to support high-mobility \gls{LEO} satellite links.

To address this challenge, we develop a dynamic-aware beam-tracking framework for \gls{LC}-\glspl{DMA} in vehicle-to-\gls{LEO} satellite uplink communications that explicitly accounts for both mobility and \gls{LC} hardware dynamics. The main contributions of this work are summarized as follows:
\begin{itemize}
    \item We establish an analytical framework for \gls{LC}-based \gls{DMA} \glspl{Tx} in satellite uplink channels that jointly models guided-wave propagation and the physical rise/decay transition dynamics of the \gls{LC} phase shifters in Section~\ref{sec: System, DMA and Dynamic LC Models}.
    \item We formulate a joint digital beamforming and dynamic-aware \gls{LC} phase-reconfiguration problem to maximize the cumulative received \gls{SNR} over a predictive tracking horizon, subject to the physical dynamics and phase limits of the \gls{LC} unit cells and the transmit power constraint in Section~\ref{sec: Problem Formulation and Solution}.
    \item We propose an \gls{AO} algorithm that alternates between a parallelized, closed-form dynamic-aware phase-shift optimization with linear complexity in the number of \gls{DMA} radiating elements in Section~\ref{sec: DMA phase shift optimization}, and a closed-form digital precoder and combiner in Section~\ref{sec: Tx Digital Precoder and Rx Combiner Optimization}.
    \item In Section~\ref{sec: Performance Comparison}, numerical simulations show that, when the time response of \gls{LC}-\glspl{DMA} is properly accounted for, they can support practical mobility scenarios. Moreover, the proposed dynamic-aware design maintains stable \gls{SNR} and outperforms conventional designs, e.g., \cite{Kimaryo2023}, particularly at high vehicle velocities and low orbital azimuth angles. 
\end{itemize}

\textit{Notation:} Vectors and matrices are denoted by boldface lowercase ($\bx$) and boldface uppercase letters ($\bX$), respectively, while scalars are represented by standard italic letters ($x$). The sets of real and complex numbers are denoted by $\Rset$ and $\Cset$, respectively, with $\Rset^N$ and $\Cset^{M \times N}$ representing the spaces of $N$-dimensional real vectors and $M \times N$ complex matrices. The transpose, Hermitian, and Kronecker-product operators are denoted by $(\cdot)^\Trans$, $(\cdot)^\Herm$, and $\otimes$, respectively. The Euclidean norm of vector $\bx$ is denoted by $\Vert{}\bx\Vert{}$, $\vert{}x\vert{}$, and $x^*$ denote the absolute and conjugate values of scalar $x$, respectively. 
The operator $\diag(\bx)$ forms a diagonal matrix with the elements of vector $\bx$ along its main diagonal, and $[\bx]_n$ denotes the $n$-th element of vector $\bx$. The $N \times N$ identity matrix, the $N \times 1$ all-zeros vector, and the $N \times 1$ all-ones vector are denoted by $\bI_N$, $\bzero_N$, and $\bone_N$, respectively. The statistical expectation is represented by $\Ex\{\cdot\}$, and $\bx \sim \sCN(\bmu, \bSigma)$ denotes a circularly symmetric complex Gaussian random vector with mean $\bmu$ and covariance matrix $\bSigma$. Finally, $\ln(\cdot)$ and $\mathcal{O}(\cdot)$ denote the natural logarithm and big-O notation, respectively.

\begin{figure*}[!t]
    \centering
    \includegraphics[width=1\textwidth]{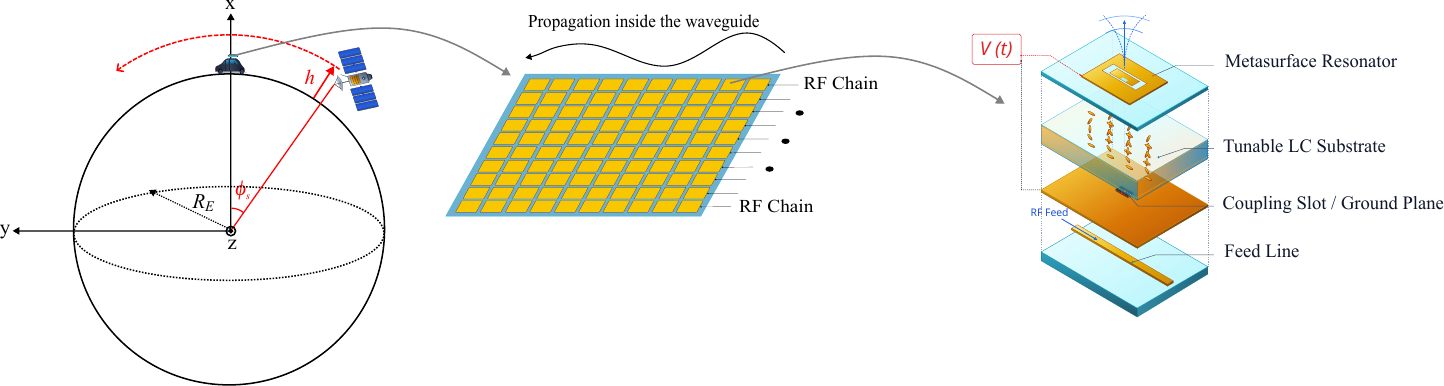}
    \caption{System model for satellite-to-vehicle communication using an \gls{LC}-based \gls{DMA} \gls{Tx}: (left) orbital geometry and communication link, (middle) waveguide-fed \gls{DMA} aperture connected to \gls{RF} chains, and (right) large view of a voltage-tunable ($V(t)$) \gls{LC} metasurface unit cell.}
    \label{fig: system model}
\end{figure*}

\section{System, DMA and Dynamic LC Models}
\label{sec: System, DMA and Dynamic LC Models}
In this section, we first present the considered system model, followed by the channel model. Finally, we discuss \gls{LC}-based \gls{DMA} time response model.
\subsection{System model}
We consider a narrowband uplink communication scenario where a moving vehicle equipped with an \gls{LC}-based \gls{DMA} transmits data to a \gls{LEO} satellite. The satellite is equipped with an $N_\rx$-element antenna array. The \gls{DMA} \gls{Tx} consists of $L$ microstrip/waveguide feed lines (each connected to a dedicated \gls{RF} chain) driving a total of $N_\tx = \sum_{\ell=1}^L N_\ell$ metamaterial radiating elements, where $N_\ell$ denotes the number of elements coupled to the $\ell$-th feed line. The received baseband signal $\by(t) \in \Cset^{N_\rx}$ at the satellite is modeled as
\begin{equation}
    \label{eq: system model}
    \by(t)=\bH(t) \bGamma(t)\bx(t)+\bn(t),
\end{equation}
where $\bH(t) \in \Cset^{N_\rx \times N_\tx}$ represents the wireless channel matrix from the $N_\tx$ metasurface elements to the satellite array, $\bx(t) \in \Cset^{L}$ is the transmit signal vector, and $\bn(t) \sim \sCN(\bzero_{N_\rx}, \sigma_n^2 \bI_{N_\rx}),\,\forall t$, denotes the \gls{AWGN} at the satellite \gls{Rx} with variance $\sigma_n^2$. The transmit signal vector is expressed as $\bx(t) = \bq s(t)$, where $\bq \in \Cset^{L}$ is the baseband digital beamformer satisfying $\|\bq\|^2 \leq P_t$ (with $P_t$ being the maximum transmit power), and $s(t) \in \Cset$ is the transmitted symbol with normalized power $\Ex\{|s(t)|^2\} = 1$. The spatial-reconfiguration and feed-waveguide response of the \gls{DMA} is represented by the block diagonal matrix $\bGamma(t) \in \Cset^{N_\tx \times L}$:
\begin{equation}
    \label{eq: phase shift channel}
    \bGamma(t)=\left[\begin{matrix}
        \btheta_{1}(t) & \bzero_{N_{1}} & \cdots & \bzero_{N_{1}}\\
        \bzero_{N_{2}} & \btheta_{2}(t) & \cdots & \bzero_{N_{2}}\\
        \vdots & \vdots & \ddots & \vdots \\
        \bzero_{N_{L}} & \bzero_{N_{L}} & \cdots & \btheta_{L}(t)
    \end{matrix}\right],
\end{equation}
where $\btheta_\ell(t) = [[\btheta_\ell]_1(t), \dots, [\btheta_\ell]_{N_\ell}(t)]^\Trans \in \Cset^{N_\ell}$ contains the effective radiating weights along feed line $\ell \in \{1, \dots, L\}$. The response of the $n_\ell$-th element on feed $\ell$ where $n_\ell\in\{1,\cdots,N_\ell\}$ is given by
\begin{equation}
    [\btheta_{\ell}]_{n_\ell}(t)=[\bm_\ell]_{n_\ell}(t)[\bzeta_\ell]_{n_\ell},
\end{equation}
where $[\bm_\ell]_{n_\ell}(t)$ is the tunable response induced by the \gls{LC} state, and $[\bzeta_\ell]_{n_\ell}$ is the waveguide response, which is given by
\begin{equation}
    \label{eq: zeta}
    [\bzeta_\ell]_{n_\ell}=\e^{-([\balpha_\ell]_{n_\ell}+\jj[\bbeta_\ell]_{n_\ell})\|\bu_{\ell,n_\ell}-\bu_{\ell,0}\|}, \forall \ell,n_\ell,
\end{equation}
where $[\balpha_\ell]_{n_\ell}$ is the waveguide attenuation coefficient (accounting for dielectric/ohmic losses and the average energy radiation from previous elements\footnote{This constant-decay approximation remains accurate due to weak element-to-waveguide coupling \cite{smith2017}.}), $[\bbeta_\ell]_{n_\ell}$ is the guided-wave propagation constant, $\bu_{\ell,n_\ell} \in \Rset^3$ is the position of the $n_\ell$-th radiator, and $\bu_{\ell,0} \in \Rset^3$ represents the feed origin of line $\ell$.

\subsection{Mobility-Induced Channel Dynamics}
\label{sec: Channel model}
Given the typical operating environment of satellite communications, the propagation channel is dominated by the far-field \gls{LoS} path. Assuming \glspl{UPA} at both the vehicle \gls{Tx} and the satellite \gls{Rx}, the channel matrix is modeled as
\begin{equation}
\label{eq: channel model}
\bH(t)=h_0\ba_\rx(\bPsi_\rx(t))\ba_\tx^\Herm(\bPsi_\tx(t)),
\end{equation}
where $h_0 \in \Cset$ denotes the complex channel attenuation factor of the \gls{LoS} link (capturing free-space path loss), while $\ba_\tx(\bPsi_\tx) \in \Cset^{N_\tx}$ and $\ba_\rx(\bPsi_\rx) \in \Cset^{N_\rx}$ represent the array response (steering) vectors of the \gls{Tx} and \gls{Rx}, respectively. The term $\bPsi_\tx = (\theta_\tx, \phi_\tx)$ specifies the \gls{AoD} toward the satellite relative to the \gls{Tx} array coordinate system, with $\theta_\tx$ and $\phi_\tx$ denoting the elevation and azimuth angles, respectively. Correspondingly, $\bPsi_\rx = (\theta_\rx, \phi_\rx)$ denotes the \gls{AoA} from the vehicle defined within the satellite \gls{Rx} array local coordinate system. Due to the motion of both the vehicle and the \gls{LEO} satellite, $\bPsi_\tx(t)$ and $\bPsi_\rx(t)$ vary over time, which makes  channel matrix $\bH(t)$ time-varying and necessitates adaptive beamforming and beam tracking.


\subsection{LC Hardware Dynamics}
\label{sec: LC time model}
In contrast to the mobility-induced channel dynamics, the \gls{LC}-\gls{DMA} exhibits internal hardware dynamics that determine how quickly its radiating elements can be reconfigured. The response of the $n_\ell$-th metamaterial element coupled to the $\ell$-th feed line of a \gls{DMA}, denoted by $[\bm_\ell]_{n_\ell}(V,t)$, is governed by a Lorentzian resonance profile \cite{smith2017,Park2024}:
\begin{equation}
    \label{eq: DMA phase shift}
    [\bm_\ell]_{n_\ell}(V,t)=\frac{\Omega f}{\Gamma}\left(\frac{f}{\Delta+\jj f}\right),
\end{equation}
where $\Delta\defeq\frac{f_0(V,t)^2-f^2}{\Gamma}$. Here, $\Omega \in \Rset$ is the oscillator strength, $f_0$ is the time-varying resonance frequency tuned by the applied bias voltage $V$, and $\Gamma$ denotes the damping factor. For \gls{LC} elements, the temporal evolution of $f_0(V,t)$ under an applied voltage $V$ can be experimentally characterized \cite{Yoo2025,Xu2024}. Moreover, in a \gls{DMA}, the amplitude and phase are inherently coupled. Let $[\bomega_\ell]_{n_\ell}(t)$ denote the phase of $[\bm_\ell]_{n_\ell}(t)$. Under resonance-tuning, narrow-band, and passivity ($\frac{\Omega f}{\Gamma}=1$) conditions, the element response at the operating carrier frequency can be parametrized in terms of its phase response $[\bomega_\ell]_{n_\ell}(t)$ as \cite{smith2017}
\begin{align}
\label{eq: m and omega}
    [\bm_\ell]_{n_\ell}(t)=\sin([\bomega_\ell]_{n_\ell}(t))\e^{\jj[\bomega_\ell]_{n_\ell}(t)},
\end{align}
where $[\bomega_\ell]_{n_\ell}(t) \in [0, \pi]$ represents the controllable phase shift of the $(\ell,n_\ell)$-th element at time $t$.

Due to the finite reorientation speed of the \gls{LC} molecules under applied electric fields, the phase shift cannot switch instantaneously. The transition model for the \gls{LC} phase shifts for the $(\ell,n_\ell)$-th element to transition from an initial phase $[\bomega^0_\ell]_{n_\ell}$ (the known current state of the elements at the current time slot) to a desired phase $[\bomega^*_\ell]_{n_\ell}$ is given respectively for $[\bomega^*_\ell]_{n_\ell}\geq[\bomega^0_\ell]_{n_\ell}$ and $[\bomega^*_\ell]_{n_\ell}\leq[\bomega^0_\ell]_{n_\ell}$ by \cite{Delbari2026}
\begin{subequations}
\label{eq: time response}
\begin{align}
    &[\bomega_\ell]_{n_\ell}(t)\!=\![\bomega^0_\ell]_{n_\ell}\!-\!\left([\bomega^0_\ell]_{n_\ell}-\omega_\tmax\right)\!\!\big(1\!-\!\e^{\!-\frac{t}{\tau^+}}\big),\\
    &[\bomega_\ell]_{n_\ell}(t)\!=\![\bomega^0_\ell]_{n_\ell}\!-[\bomega^0_\ell]_{n_\ell}\big(1\!-\!\e^{\!-\frac{t}{\tau^-}}\big),
\end{align}
\end{subequations}
where $\tau^+$ and $\tau^-$ denote the rise (electric field-driven) and decay (relaxation-driven) reorientation time constants of the \gls{LC} material, respectively, which depend on the \gls{LC} physical parameters; see Fig.~\ref{fig: exp-tau}. Furthermore, $\omega_{\max} = \pi$ denotes the maximum achievable phase shift range. Note that for both equations in \eqref{eq: time response}, $[\bomega_\ell]_{n_\ell}(t)=[\bomega^*_\ell]_{n_\ell},\forall t,$ as soon as it reaches the desired phase shift since we used over- and undershooting technique \cite{Delbari2026}.

\begin{figure}[t]
    \centering
    \includegraphics[width=0.48\textwidth]{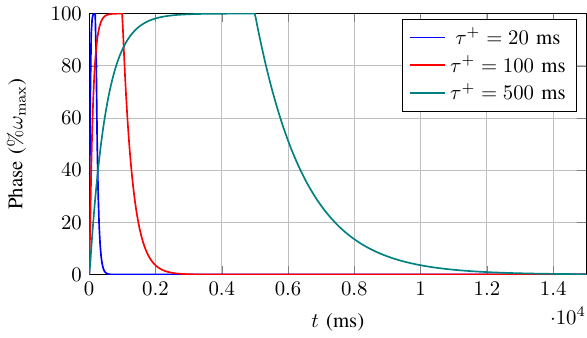}
    \caption{Transient response (rise and decay dynamics) of the \gls{LC} cell for various relaxation time constants $\tau$ resulting from differing \gls{LC} parameters, where $\tau^-=3\tau^+$ in all curves \cite{Xu2024,Delbari2026}.}
    \label{fig: exp-tau}
\end{figure}

\section{Problem Formulation and Solution}
\label{sec: Problem Formulation and Solution}
We assume that the vehicle trajectory over the prediction horizon can be predicted using existing algorithms such as \cite{Yuan2021}, while the satellite trajectory is known. Accordingly, the mobility-induced channel dynamics $\bH(t)$ in \eqref{eq: channel model} are determined by the predicted trajectories. Moreover, the \gls{LC} hardware dynamics govern the achievable \gls{DMA} configuration $\bGamma(t)$ according to \eqref{eq: time response}. As a \gls{QoS} metric, we consider the receive \gls{SNR} at the satellite, which captures both dynamics and is defined as
\begin{equation}
    \label{eq: SNR}
    \SNR(t)=\frac{|\bw^\Herm\bH(t) \bGamma(t)\bq|^2}{\sigma_n^2},
\end{equation}
where $\bw\in\Cset^{N_\rx}$ is the receive combining vector. Assuming $N_e\defeq N_1=N_2=\cdots = N_L$, we can drop the index $\ell$ from $n_\ell$ and denote it as $n=(\ell-1)N_e+n_e,\,\forall \ell=1,\ldots,L$ and $n_e=1,\ldots,N_e$. We then formulate the dynamic-aware beamforming problem over $S$ sampling instants as:
\begin{subequations}
\label{eq: optimization 1}
\begin{align}
    \quad&~\underset{\bomega,\bq,\bw}{\max}~\sum_{s=1}^{S}\SNR(t_s)
    \label{eq: optimization 1 a}\\&~\text {s.t.} ~~ 0\leq [\bomega]_{n} < \omega_\tmax, \forall n,
    \label{eq: optimization 1 b}\\&\quad\hphantom {\text {s.t.} } \|\bq\|_2^2\leq P_t, \quad \|\bw\|\leq 1,
    \label{eq: optimization 1 c}
\end{align}
\end{subequations}
where $\bomega\defeq[\bomega_1^\Trans,\cdots,\bomega_L^\Trans]^\Trans\in\Rset^{N_\tx\times 1}$ and $t_s$ is the sampling time. 
Here, \eqref{eq: optimization 1 b} is the realizable
phase-shift range with $\omega_\tmax=\pi$ being the maximum unit cell phase-shift, and \eqref{eq: optimization 1 c} enforces the maximum transmit power budget and combiner constraint. The problem \eqref{eq: optimization 1} is a non-convex problem due to the non-convexity of the objective function. Therefore, obtaining a global solution is challenging, and we optimize $\{\bw,\bq\}$ and $\bomega$ using an \gls{AO} framework.

\subsection{LC-DMA Phase-Shift Optimization}
\label{sec: DMA phase shift optimization}
Assuming a given fixed $\bq$ and $\bw$, the time reconfiguration of the \gls{DMA} phase configuration problem is:
\begin{subequations}
\label{eq: optimization 2}
\begin{align}
    \quad&~\underset{\bomega}{\max}~\sum_{s=1}^{S}|\bg^\Herm(t_s)\bm(t_s)|^2
    \label{eq: optimization 2 a}\\&~\text {s.t.} ~~ 0\leq [\bomega]_{n} < \omega_\tmax, \forall n,
    \label{eq: optimization 2 b}
\end{align}
\end{subequations}
where $\bm(t)\defeq[\bm_1^\Trans(t),\cdots,\bm_L^\Trans(t)]^\Trans\in\Cset^{N_\tx}$ with $[\bm_\ell]_{n_\ell}(t),\forall \ell,n_\ell$ defined in \eqref{eq: m and omega}, and 
\begin{equation}
\label{eq: g equation}
\bg^\Herm(t)\defeq\left(\bw^\Herm\bH(t)\diag(\bzeta)\bQ\right)/\sigma_n,
\end{equation}
with $\bzeta\defeq[\bzeta_1^\Trans,\cdots,\bzeta_L^\Trans]^\Trans\in\Cset^{N_\tx}$ with $[\bzeta_\ell]_{n_e},\,\forall \ell,n_e$ defined in \eqref{eq: zeta} and $\bQ=\diag\left(\bq\otimes\bone_{N_e}\right)\in\Cset^{N_\tx\times N_\tx}$.

Problem \eqref{eq: optimization 2} is non-convex due to the non-linear coupling of $\bomega$ and time in \eqref{eq: optimization 2 a}. To solve this problem efficiently in real time with linear complexity, we first decouple the objective over the individual metamaterial elements across time snapshots in the following Lemma.
\begin{lem}
\label{lem}
    For sufficiently large $N_\tx$ and under the phase alignment condition $\arg\left(\bg^\Herm(t_s)\bm(t_s)\right)\approx 0$, the spatial SNR metric at snapshot $s$ can be decoupled across individual elements $n$ as, $|\bg^\Herm(t_s)\bm(t_s)|^2$
    \begin{equation}
         \approx \sum_{n=1}^{N_\tx} 2 r_{n}(t_s) \sin([\bomega]_n(t_s))\cos([\bomega]_n(t_s) - [\bvarphi]_n(t_s)) + c(t_s),
    \end{equation}
    where $[\bvarphi]_n(t_s) = \arg([\bg]_n(t_s))$, $r_n(t_s) = \big|\sum_{m \neq n} \sin([\bomega]_m(t_s))\e^{\jj(-[\bomega]_m(t_s) + [\bvarphi]_m(t_s))}\big|$, and $c(t_s)$ is a constant independent of $[\bomega]_n(t_s)$.
\end{lem}
\begin{proof}
    Expanding $|\bg^\Herm(t_s)\bm(t_s)|^2$ around element $p$ while fixing all other elements $m \neq p$ yields, $|\bg^\Herm(t_s)\bm(t_s)|^2$
    \begin{align}
        \label{eq: SNR expansion}
         = &\sin([\bomega]_p)\e^{\jj([\bomega]_p-[\bvarphi]_p)}A(p) + \sin([\bomega]_p)\e^{-\jj([\bomega]_p-[\bvarphi]_p)}A^*(p) \nonumber \\
        &+ \sin^2([\bomega]_p) + |A(p)|^2,
    \end{align}
    where $A(p) = \sum_{m\neq p}\sin([\bomega]_m)\e^{\jj(-[\bomega]_m+[\bvarphi]_m)}$. Letting $A(p) = r_p \e^{\jj \arg(A(p))}$ and observing that $\arg(A(p)) \approx 0$ under coherent alignment, \eqref{eq: SNR expansion} simplifies to $\sin^2([\bomega]_p) + 2 r_p \sin([\bomega]_p)\cos([\bomega]_p - [\bvarphi]_p) + |A(p)|^2$. For large arrays ($N_\tx \gg 1$), the cross-term dominates the single-element self-radiation term $\sin^2([\bomega]_p)$, which completes the proof.
\end{proof}

\begin{algorithm}[t]
\caption{Dynamic-Aware Predictive \gls{LC}-\gls{DMA} Tracking}
\label{alg:AO_DMA}
\begin{algorithmic}[1]
\STATE \textbf{Input:} Current phase state $\bomega^0$, sampling duration $\Delta t$, horizon $S$, $\tau^+$, and $\tau^-=3\tau^+$.
\STATE Predict look-ahead array responses $\{\ba_\tx(t_s), \ba_\rx(t_s)\}_{s=1}^S$ and construct channel states $\bH(t_s)$.
\STATE Compute temporal weights $w_s(\tau),\,\forall s\!=\!1,\dots,S$ via \eqref{eq: weights definition}.
\STATE Compute the ideal instantaneous phases $[\bomega^{\mathrm{inst}}]_n(t_s) = \frac{1}{2}\arg([\bg]_n(t_s)) + \frac{\pi}{4}, \forall n, s$.
\STATE Compute the commanded phase vector $\bomega^\star$ in parallel across all $N_\tx$ elements via \eqref{eq: weighted phase update}.
\STATE Construct metasurface matrix $\bGamma$ using current continuous phase state and update digital precoder $\bq$ via \eqref{eq: q beamformer}.
\STATE Update physical \gls{LC} relaxation states for subsequent time steps according to \eqref{eq: time response}.
\end{algorithmic}
\end{algorithm}

Using Lemma~\ref{lem}, maximizing \eqref{eq: optimization 2 a} over the prediction horizon $S$ separates into $N_\tx$ parallel subproblems. For each element $n$, the optimization task is formulated as
\begin{equation}
    \label{eq: element subproblem}
    \underset{[\bomega]_n^\star}{\max} \sum_{s=1}^S \sin([\bomega]_n(t_s))\cos([\bomega]_n(t_s) - [\bvarphi]_n(t_s)),
\end{equation}
where $[\bomega]_n(t_s)$ is the continuous physical phase reached at time $t_s = s\Delta t$ when transitioning from the initial state $[\bomega^0]_n$ toward the commanded phase $[\bomega]_n^\star$ according to the relaxation trajectory in \eqref{eq: time response}.
For any instantaneous snapshot $s$, the unconstrained maximizer of $\sin(x)\cos(x - \phi_s)$ is given in closed form by $x_s^\star = \frac{\pi}{4} + \frac{\phi_s}{2}$.
Because the physical phase trajectory cannot switch instantaneously, the commanded phase target $[\bomega]_n^\star$ must steer the element toward the optimal future phase profile while respecting the finite response speed of the \gls{LC} molecules. Over a transition window driven primarily by electric-field excitation ($\{\tau^+,\tau^-\}$), the unit step response of the \gls{LC} phase at step $s$ exhibits an effective factor of $ 1 - \exp(-s\Delta t / \tau),\,\forall \tau\in\{\tau^+,\tau^-\}$. Consequently, applying a second-order Taylor tracking approximation around the per-step optima $\{x_{n,s}^\star\}_{s=1}^S$ yields the optimal commanded phase shift in closed form:
\begin{equation}
    \label{eq: weighted phase update}
    [\bomega^\star]_n \!\!= \!\!\!\!\!\!\!\mod\!\!\!\!\left(\arg \left( \sum_{s=1}^S w_s(\tau) \exp\left( \jj \left( \frac{[\bvarphi]_n(t_s)}{2} + \frac{\pi}{4} \right)\!\! \right)\!\! \right)\!,\pi\right),
\end{equation}
where the dynamic-aware temporal weights are given by
\begin{equation}
    \label{eq: weights definition}
    w_s(\tau) = \frac{1 - \exp\left(-\frac{s\Delta t}{\tau}\right)}{\sum_{j=1}^S \left(1 - \exp\left(-\frac{j\Delta t}{\tau}\right)\right)},\,\forall \tau\in\{\tau^+,\tau^-\}.
\end{equation}
The closed-form update in \eqref{eq: weighted phase update}--\eqref{eq: weights definition} optimizes all $N_\tx$ elements in parallel with $\mathcal{O}(1)$ complexity per unit cell.

\subsection{Tx Digital Precoder and Rx Combiner Optimization}
\label{sec: Tx Digital Precoder and Rx Combiner Optimization}
For a fixed phase-shift profile $\bomega$, Problem~\eqref{eq: optimization 1} reduces to
\begin{subequations}
\label{eq: optimization 4}
\begin{align}
    \quad&~\underset{\bq,\bw}{\max}~\sum_{s=1}^{S}\SNR(t_s)
    \label{eq: optimization 4 a}\\&~\text {s.t.} ~~ \|\bq\|_2^2\leq P_t, \quad \|\bw\|\leq 1,
    \label{eq: optimization 4 b}
\end{align}
\end{subequations}
which optimizes the baseband precoder $\bq$ and receive combiner $\bw$ under the transmit power budget. The optimal combiner and precoder that maximize the received summation over \gls{SNR} are given by the \gls{AO} over $\bq$ and $\bw$ \cite{tse2005fundamentals}:
\begin{align}
\label{eq: q beamformer}
\bq &= \sqrt{P_t} \be_{\max}\left(\sum_{s=1}^{S}\bh_q(t_s)\bh_q^\Herm(t_s)\right), \\ \bw &= \be_{\max}\left(\sum_{s=1}^{S}\bh_w(t_s)\bh_w^\Herm(t_s)\right),
\end{align}
where $\bh_q^\Herm(t_s)\defeq \bw^\Herm\bH(t_s)\bGamma(t_s)$, $\bh_w(t_s)\defeq \bH(t_s)\bGamma(t_s)\bq$ and $\be_{\max}(\cdot)$ denotes the unit-norm eigenvector associated with the maximum eigenvalue of the matrix.

\begin{figure}[!t]
    \centering \includegraphics[width=0.48\textwidth]{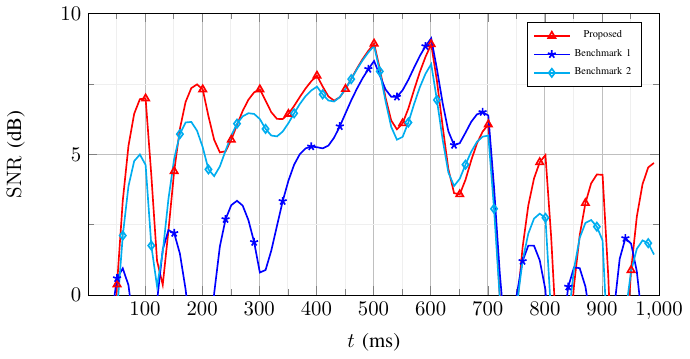}
    \caption{Achieved \gls{SNR} versus time for satellite azimuth angles $\phi_s \approx -24^\circ$ at a vehicle speed of $v = 15$\,m/s where $\tau^+=100$ ms and $\tau^-=3\tau^+$.}
    \label{fig:SNR-time}
\end{figure}

\subsection{Algorithm and Complexity Analysis}
The overall procedure is summarized in Algorithm~\ref{alg:AO_DMA}. In each reconfiguration period, the look-ahead channel steering vectors are predicted across the horizon $S$. Each element's target phase is then updated in parallel via the closed-form dynamic-aware estimator in \eqref{eq: weighted phase update}. Finally, the digital precoder is updated via \eqref{eq: q beamformer}. Evaluating \eqref{eq: weighted phase update} across all $N_\tx$ elements requires $\mathcal{O}(S N_\tx)$ operations, and computing $\bq$ requires $\mathcal{O}(N_\rx N_\tx + L)$. Thus, the proposed framework exhibits a strictly linear computational complexity of $\mathcal{O}\left(N_\tx \left(S + N_\rx\right)\right)$, making it well suited for real-time tracking in high-mobility satellite uplinks.

\begin{figure}[!t]
    \centering
    \includegraphics[width=0.48\textwidth]{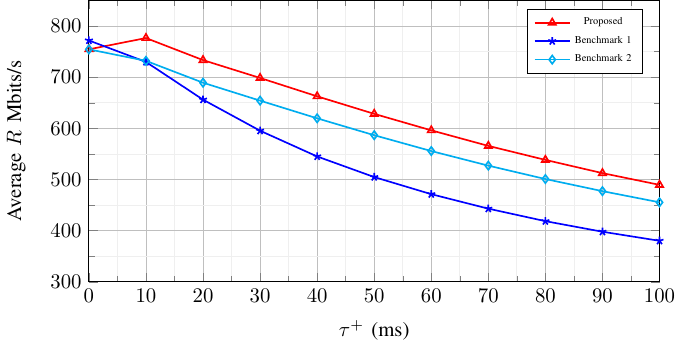}
    \caption{Average data rate versus reorientation time constants of the \gls{LC} materials when $\tau^-=3\tau^+$. $\phi_s \approx -24^\circ$ and $v = 15$\,m/s.}
    \label{fig:Data taup}
\end{figure}

\section{Performance Evaluation}
\label{sec: Performance Comparison}

\textbf{Simulation Setup:} We consider an Earth-centered coordinate system where the origin is at the center of the Earth and the Earth's radius is $R_E = 6371$\,km, as illustrated in Fig.~\ref{fig: system model}. The ground vehicle is initially located at $[R_E, 0, 0]^\Trans$ and travels at a speed of $v$\,m/s along a local circular path of radius $r_u = 5$\,m. We assume a \gls{LEO} satellite orbiting in the $x$-$y$ plane with a constant orbital speed of $7.5$\,km/s at an altitude of $h = 600$\,km. The trajectory of the satellite is described in Cartesian coordinates as $[(R_E + h)\cos(\phi_s), (R_E + h)\sin(\phi_s), 0]^\Trans$, where $\phi_s$ denotes the orbital azimuth angle. The satellite is equipped with a $20 \times 20$ element \gls{UPA} with inter-element spacing $d_\rx = \lambda/2$ operating at a carrier frequency of $f = 14$\,GHz, where $\lambda = c/f$ is the wavelength and $c$ is the speed of light in vacuum. The vehicle \gls{Tx} employs an \gls{LC}-based \gls{DMA} composed of $L = 10$ feed lines, each driving $N_e = 20$ radiating elements with an element spacing of $\lambda/6$. The guided-wave parameters are configured as $[\balpha_\ell]_{n_e} = 0.6$\,Np/m and $[\bbeta_\ell]_{n_e} = 1.4 \times (2\pi/\lambda)$ for all $\ell$ and $n_e$ \cite{Sun2026}. The \gls{Rx} noise power is calculated as $\sigma_n^2 = W N_0 N_f$, where $N_0 = -174$\,dBm/Hz is the \gls{AWGN} power spectral density, $W = 250$\,MHz is the channel bandwidth, and $N_f = 6$\,dB represents the \gls{Rx} noise figure. 

\begin{figure}[t]
    \centering    \includegraphics[width=0.48\textwidth]{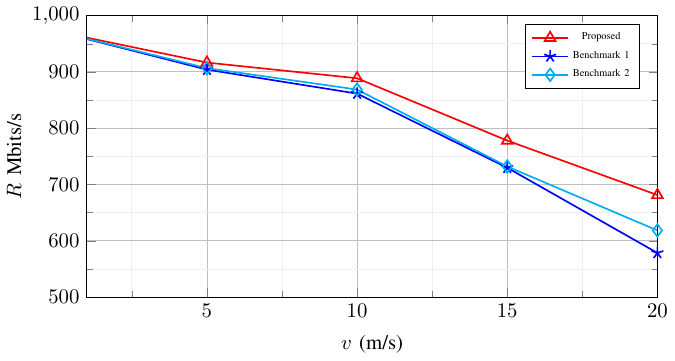}
    \caption{Achievable data rate versus vehicle speed for satellite azimuth angle $\phi_s \approx -24^\circ$.}
    \label{fig:vehicle}
\end{figure}

\textbf{Benchmark Schemes:} We compare the proposed method, which accounts for both channel and \gls{LC} dynamics, against two benchmarks. \textit{Benchmark 1 (neglecting channel dynamics):} Uses only the instantaneous channel state ($S=1$), without predicting mobility-induced channel evolution. \textit{Benchmark 2 (neglecting \gls{LC} dynamics):} Predicts the channel evolution over the horizon $S$, but designs the target phase shifts under the naive assumption of instantaneous \gls{LC} reconfiguration (i.e., setting uniform weights $w_s = 1/S$ and ignoring switching dynamics by considering $\tau^+ = \tau^- = 0$ during optimization). During evaluation, however, the phase transitions are executed subject to the realistic physical \gls{LC} time response in \eqref{eq: time response}.

\textbf{Simulation Results:} Fig.~\ref{fig:SNR-time} illustrates the achieved \gls{SNR} over time for satellite azimuth angle $\phi_s = -24^\circ$. 
The dynamic phase reconfiguration is necessary to track the channel variations. While the Benchmark~1 scheme focuses solely on maximizing the instantaneous \gls{SNR} without accounting for the \gls{LC} switching delay, and Benchmark~2 assumes \gls{LC} phase shifts can change instantly, the proposed transition-aware algorithm optimizes phase transitions that maximize the cumulative received \gls{SNR} over the tracking horizon by considering the time response constraint of the \gls{LC}. Consequently, the benchmark methods fail to attain the desired \gls{SNR} within the channel coherence window of the moving vehicle, whereas the proposed method maintains stable connectivity. Fig. \ref{fig:Data taup} illustrates the average achieved data rate when the orientation time constants of the \gls{LC} material change. As $\tau^+$ increases (with $\tau^- = 3\tau^+$ \cite{Xu2024}), the performance decreases as the tracking becomes more challenging for \gls{LC}.

Figs.~\ref{fig:vehicle} and \ref{fig:satellite} evaluate the achievable data rate for the proposed transition-aware and the benchmark approaches. Fig.~\ref{fig:vehicle} shows that the performance gain of the proposed scheme becomes more pronounced as the vehicle velocity increases, as it penalizes slow phase transitions. Furthermore, Fig.~\ref{fig:satellite} illustrates the achievable rate across various satellite positions. As the satellite moves toward lower orbital elevation/azimuth angles relative to the vehicle, the required spatial phase gradients change more rapidly, thereby magnifying the necessity and performance advantage of the proposed algorithm.
\begin{figure}[t]
    \centering
    \includegraphics[width=0.48\textwidth]{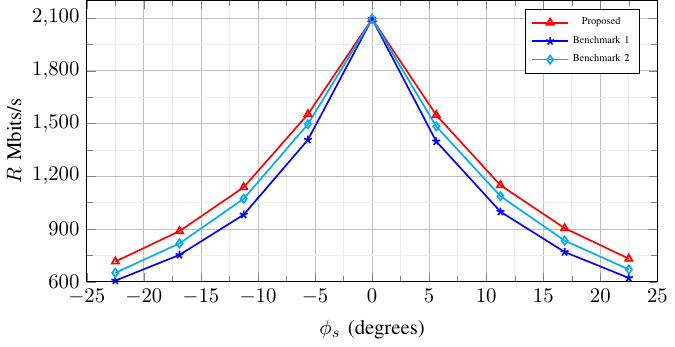}
    \caption{Achievable data rate versus satellite position at a vehicle speed of $v = 20$\,m/s.}
    \label{fig:satellite}
\end{figure}
    
\section{Conclusion}
\label{sec: conclusion}
In this paper, we addressed beam tracking for high-mobility vehicle-to-\gls{LEO} satellite uplinks using \gls{LC}-based \glspl{DMA}. We proposed a dynamic-aware predictive beamforming framework that jointly accounts for mobility-induced channel dynamics and \gls{LC} hardware dynamics. The proposed solution combines predictive beam tracking with dynamic-aware \gls{LC} phase reconfiguration and achieves linear computational complexity. Simulation results demonstrate more stable \gls{SNR} and higher achievable rates than schemes neglecting either channel or \gls{LC} dynamics, particularly in high-mobility scenarios.

\section*{Acknowledgment}
The authors would like to thank Prof. Anja Klein for her valuable comments on this work.

\bibliographystyle{IEEEtran}
\bibliography{References}

\end{document}